\documentclass{article}
\usepackage{spconf,amsmath,graphicx}

\usepackage{cite}

\usepackage{amssymb}
\usepackage{amsthm}

\usepackage{cite}
\usepackage{array}
\usepackage{url,color}

\def\kron{\otimes}

\def\diag{\mathrm{diag}}

\def\Htran{\mbox{\tiny H}}
\def\Ttran{\mbox{\tiny T}}

\newcommand{\fracSumtwo}[2]{\overset{#2}{\underset{#1}{\sum}}}
\newcommand{\vect}[1]{\mathbf{#1}}

\theoremstyle{remark}

\newtheorem{lemma}{Lemma}

\title{Circuit-Aware Design of Energy-Efficient Massive MIMO Systems}

\name{Emil~Bj\"ornson$^{\star \dagger}$ \qquad Michail Matthaiou$^{\ddagger}$ \qquad M\'erouane~Debbah$^{\star}$\thanks{E.~Bj\"ornson is funded by the International Postdoc Grant 2012-228 from the Swedish Research Council. This research has been supported by the ERC Starting Grant 305123 MORE (Advanced Mathematical Tools for Complex Network Engineering).}}

\address{\normalsize $^\star$Alcatel-Lucent Chair on Flexible Radio, SUPELEC, Gif-sur-Yvette, France\\
\normalsize$^\dagger$ACCESS Centre, Dept.~of Signal Processing, KTH Royal Institute of Technology, Stockholm, Sweden\\
\normalsize$^\ddagger$ECIT Institute, Queen's University Belfast, Belfast, U.K. and S2, Chalmers University of Technology, Gothenburg, Sweden}

\begin{document}

\maketitle

\begin{abstract}
Densification is a key to greater throughput in cellular networks. The full potential of coordinated multipoint (CoMP) can be realized by massive multiple-input multiple-output (MIMO) systems, where each base station (BS) has very many antennas. However, the improved throughput comes at the price of more infrastructure; hardware cost and circuit power consumption scale linearly/affinely with the number of antennas. In this paper, we show that one can make the circuit power increase with only the square root of the number of antennas by circuit-aware system design. To this end, we derive achievable user rates for a system model with hardware imperfections and show how the level of imperfections can be gradually increased while maintaining high throughput. The connection between this scaling law and the circuit power consumption is established for different circuits at the BS.
\end{abstract}

\section{Introduction}
\label{sec:intro}

We consider a cellular network where each BS communicates with $K$ unique single-antenna user equipments (UEs). Interference coordination is a major limiting factor in these systems but can be handled in the spatial domain by CoMP methods, where several antennas, $N$, are deployed at each BS \cite{Gesbert2010a}. The massive MIMO paradigm, where $N \gg K$, has gained particular traction in recent years \cite{Rusek2013a}, because it allows for distributed interference coordination and brings robustness to having imperfect channel state information (CSI).

Two important practical issues with the deployment of large antenna arrays are the increased hardware cost and circuit power consumption \cite{Bjornson2014a}---these scale linearly/affinely with $N$ unless we redesign the network with these two issues in mind. Low-cost energy-efficient transceiver equipment suffer from hardware imperfections, which must be modeled properly if accurate conclusions are to be drawn \cite{Schenk2008a,Mezghani2010a,Bjornson2014c}. In this paper, we consider an uplink system distorted by multiplicative phase-drifts, additive distortion noise, and noise amplifications. We derive achievable user rates and prove that the level of imperfections can be gradually increased with $N$. The practical implications of this scaling law are established for three circuits at the BSs:  analog-to-digital converter (ADC), low noise amplifier (LNA), and local oscillator (LO). These are the main components of the typical receiver illustrated in Fig.~\ref{figure_blockdiagram}.

\begin{figure}
\begin{center}
\includegraphics[width=.9\columnwidth]{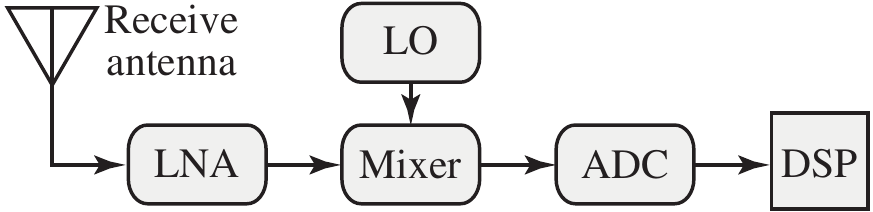}
\end{center}\vskip-7mm
\caption{Block diagram of one antenna branch in a typical receiver. The main hardware components are given, while various filters are also needed.} \label{figure_blockdiagram} \vskip-2mm
\end{figure}

\section{System Model}
\label{sec:system-model}

We consider the uplink of a network with $L\geq 1$ cells. The flat-fading channel from UE $k$ in cell $l$ to BS $j$ is denoted as $\vect{h}_{jlk} \triangleq [h_{jlk}^{(1)} \, \ldots \, h_{jlk}^{(N)}]^{\Ttran} \in \mathbb{C}^N$ and is modeled as block fading; thus, it is static for a coherence block of $T$ channel uses and has independent realizations between blocks. Each channel is zero-mean circularly symmetric complex Gaussian distributed as $\vect{h}_{jlk} \sim \mathcal{CN}(\vect{0},\lambda_{jlk} \vect{I}_N)$, where the average channel attenuation $\lambda_{jlk}>0$ depends on the large-scale fading.

The signal $x_{lk}(t)$ sent by UE $k$ in cell $l$ at channel use $t$ satisfies a power constraint of $\mathbb{E}\{ |x_{lk}(t) |^2 \} = p_{lk}$ and $\vect{x}_{l}(t) \triangleq [x_{l1}(t) \, \ldots \, x_{lK}(t)]^{\Ttran} \in \mathbb{C}^{K}$ is the transmit signal in cell $l$.

\begin{figure*}[t!]
\begin{align} \label{eq:achievable-SINR}
\mathrm{SINR}_{jk}(t) =  \frac{ p_{jk} | \mathbb{E}\{ \vect{v}_{jk}^{\Htran}(t) \vect{h}_{jjk}(t) \} |^2 }{ \fracSumtwo{l=1}{L} \fracSumtwo{m=1}{K} p_{lm}  \mathbb{E}\{ |\vect{v}_{jk}^{\Htran}(t) \vect{h}_{jlm}(t) |^2  \} - p_{jk} | \mathbb{E}\{ \vect{v}_{jk}^{\Htran}(t) \vect{h}_{jjk}(t) \} |^2 + \mathbb{E}\{ |\vect{v}_{jk}^{\Htran}(t) \boldsymbol{\upsilon}_j(t) |^2  \}   + \sigma^2 \xi \mathbb{E}\{ \| \vect{v}_{jk}(t) \|^2\}  } \tag{7}
\end{align} \vskip-2mm
\hrulefill
\vskip-4mm
\end{figure*}

Contrary to most prior works on massive MIMO (e.g., \cite{Rusek2013a} and references therein), the receiver branches at the BSs are assumed to be imperfect. Interestingly, it is shown in \cite{Schenk2008a,Mezghani2010a,Bjornson2014c} that imperfect hardware mainly causes
multiplicative phase-drifts, additive distortion noise, and noise amplifications. Based on these prior works, the received signal $\vect{y}_j(t) \in \mathbb{C}^{N}$ at BS $j$ at channel use $t \in \{ 1,\ldots,T \}$ in the coherence block is modeled as \vspace{-3mm}
\begin{equation} \label{eq:generalized-model}
\vect{y}_j(t) = \vect{D}_{\boldsymbol{\phi}_j(t)} \sum_{l=1}^{L} \vect{H}_{jl} \vect{x}_{l}(t) + \boldsymbol{\upsilon}_{j}(t) + \boldsymbol{\eta}_{j}(t)
\end{equation}
where $\vect{H}_{jl} \triangleq [ \vect{h}_{jl1} \, \ldots \, \vect{h}_{jlK}] \in \mathbb{C}^{N \times K}$ is the channel from UEs in cell $l$, while the hardware imperfections are given by:
\begin{enumerate}

\item Phase-drift matrix $\vect{D}_{\boldsymbol{\phi}_j(t)} \! \triangleq \! \diag(e^{\imath \phi_{j1}(t)},\ldots, e^{\imath \phi_{jN}(t)})$ where the drift at the $n$th antenna of BS $j$ at time $t$ follows a Wiener process: $\phi_{jn}(t) \!\sim\! \mathcal{N}( \phi_{jn}(t\!-\!1), \delta)$. The parameter $\delta \geq 0$ is the variance of the innovations.

\item Additive distortion noise $\boldsymbol{\upsilon}_{j}(t) \! \sim \! \mathcal{CN}(\vect{0},\vect{\Upsilon}_j )$ which is independent across time and antennas, such that
\begin{equation*}
 \vect{\Upsilon}_j \triangleq \kappa^2 \sum_{l=1}^{L} \sum_{k=1}^{K} p_{lk} \diag\left( |h_{jlk}^{(1)}|^2,\ldots, |h_{jlk}^{(N)}|^2\right)\notag
\end{equation*}
for a given channel realization. The variance at an antenna is proportional to the received signal power at this antenna and $\kappa \geq 0$ is the proportionality parameter \cite{Schenk2008a}.

\item The receiver noise $\boldsymbol{\eta}_{j}(t)  \sim \mathcal{CN}(\vect{0},\sigma^2 \xi \vect{I}_N)$ where
$\sigma^2$ is the fundamental thermal noise power and the parameter $\xi \geq 1$ is the noise amplification factor.
\end{enumerate}

This tractable model of hardware imperfections at the BS is characterized by three parameters: $\delta$, $\kappa$, and $\xi$. We exemplify in Section \ref{sec:circuit-examples} how these are related to the main hardware components of the BSs which are shown in Fig.~\ref{figure_blockdiagram}. In the analysis of this paper, we distinguish between two important cases: a common LO (CLO) on all antennas of a BS and separate LOs (SLOs) with identical properties. In the former case, the phase drifts $\phi_{jn}(t)$ are identical for all $n=1,\ldots,N$, while these drifts are independent when having SLOs.

\section{Achievable User Rates and Scaling Laws}

Next, we provide our main analytic throughput results for the system model in \eqref{eq:generalized-model}. These results are exploited in Section \ref{sec:circuit-examples} where we describe the individual impact of different hardware components on massive MIMO configurations.

Although the channel is fixed within the coherence block, the effective channels $\vect{h}_{jlk}(t) \triangleq \vect{D}_{\boldsymbol{\phi}_j(t)} \vect{h}_{jlk} \in \mathbb{C}^{N}$ change due to the phase-drifts. Hence, we need to estimate the channel for each $t$ used for transmission. Suppose pilot sequences that occupy  $B $ channel uses are transmitted in the beginning of each coherence block; $\tilde{\vect{x}}_{jk} \triangleq [x_{jk}(1) \, \ldots \, x_{jk}(B)]^{\Ttran} \in \mathbb{C}^{B}$ is the pilot sequence of UE $k$ in cell $j$. The following is the linear minimum mean squared error (LMMSE) estimator.

\begin{lemma}
\label{theorem:LMMSE-estimation}
Let $\boldsymbol{\psi}_j \triangleq [\vect{y}_j^{\Ttran}(1) \, \ldots \, \vect{y}_j^{\Ttran}(B)]^{\Ttran} \in \mathbb{C}^{NB}$ denote the received signal at BS $j$ during the pilot transmission. For both a CLO and SLOs, the LMMSE estimate of $\vect{h}_{jlk}(t)$ at any channel use $t\geq B$ for any $l$ and $k$ is
\begin{equation} \label{eq:LMMSE-estimator}
  \hat{\vect{h}}_{jlk}(t) = \left( \lambda_{jlk} \tilde{\vect{x}}_{lk}^{\Htran}  \vect{D}_{\boldsymbol{\delta}(t)} \boldsymbol{\Psi}^{-1}_j \kron \vect{I}_N \right) \boldsymbol{\psi}_j
\end{equation}
where $\vect{D}_{\boldsymbol{\delta}(t)} \triangleq \diag( e^{-\frac{\delta}{2} (t-1)}, e^{-\frac{\delta}{2} (t-2)}, \ldots, e^{-\frac{\delta}{2} (t-B)})$, \vspace{-3mm}
\begin{align}
\boldsymbol{\Psi}_j &\triangleq \sum_{\ell=1}^{L} \sum_{m=1}^{K} \lambda_{j \ell m} \vect{X}_{\ell m} + \sigma^2 \xi \vect{I}_B, \\
\vect{X}_{\ell m} &\triangleq \bar{\vect{X}}_{\ell m} + \kappa^2 \vect{D}_{| \tilde{\vect{x}}_{\ell m}|^2}, \\
\vect{D}_{| \tilde{\vect{x}}_{\ell m}|^2} &\triangleq \diag( |x_{\ell m}(1)|^2 , \ldots, |x_{\ell m}(B)|^2 ),
\end{align}
$\kron$ is the Kronecker product, and element $(i_1,i_2)$ of $\bar{\vect{X}}_{\ell m}$ is
\begin{equation}
[\bar{\vect{X}}_{\ell m} ]_{i_1,i_2} = \begin{cases} |x_{\ell m}(i_1)|^2, & i_1 = i_2, \\ x_{\ell m}(i_1) x_{\ell m}^*(i_2) e^{-\frac{\delta}{2} |i_1-i_2|}, &  i_1 \neq i_2. \end{cases}
\end{equation}
\end{lemma}
\begin{proof}
This LMMSE estimator was derived for SLOs in \cite{Bjornson2014c} and the same derivations hold for a CLO.
\end{proof}

The channel estimates in Lemma \ref{theorem:LMMSE-estimation} are utilized to select receiver filters $\vect{v}_{jk}^{\Htran}(t) \in \mathbb{C}^{N}$.
Using an approach from \cite{Pitarokoilis2012a} and  \cite[Lemma 1]{Bjornson2014c}, the achievable rate for UE $k$ in cell $j$ is
\begin{equation*} \label{eq:achievable-rate}
R_{jk} = \frac{1}{T}  \sum_{t=B+1}^{T} \log_2 \big( 1 + \mathrm{SINR}_{jk}(t) \big) \quad [\textrm{bit/channel use}]
\end{equation*}
where $\mathrm{SINR}_{jk}(t)$ is given in \eqref{eq:achievable-SINR} at the top of this page.
The expectations in  \eqref{eq:achievable-SINR} can be computed in closed form if the BS applies maximum ratio combining (MRC): $\vect{v}_{jk}(t) =  \hat{\vect{h}}_{jjk}(t)$.

\setcounter{equation}{7}

\begin{lemma} \label{theorem:MRC-expectations}
If the MRC receive filter is used, then
\begin{align}
\mathbb{E}\{ \|\vect{v}_{jk}(t) \|^2\} & = N \lambda_{jjk}^2 \tilde{\vect{x}}_{jk}^{\Htran}  \vect{D}_{\boldsymbol{\delta}(t)} \boldsymbol{\Psi}^{-1}_j \vect{D}_{\boldsymbol{\delta}(t)}^{\Htran} \tilde{\vect{x}}_{jk} \notag \\
\mathbb{E}\{ \vect{v}_{jk}^{\Htran}(t) \vect{h}_{jjk}(t) \} & = \mathbb{E}\{ \|\vect{v}_{jk}(t) \|^2\}  \notag \\
\mathbb{E}\{ | \vect{v}_{jk}^{\Htran}(t) \vect{h}_{jlm}(t) |^2 \} &= \lambda_{jlm} \mathbb{E}\{ \| \vect{v}_{jk}(t) \|^2\}  \notag  \\
&\!\!\!\!\!\!\!\!\!\!\!\!\!\!\!\!\!\!\!\!\!\!\!\!\!\!\!\!\!\!\!\!\!\!\!\!\!\!\!\!\!\!\!\!\!\!\!\!\!\!\!\!  + N  \lambda_{jjk}^2 \lambda_{jlm}^2 \tilde{\vect{x}}_{jk}^{\Htran}   \vect{D}_{\boldsymbol{\delta}(t)} \boldsymbol{\Psi}_j^{-1} \vect{X}_{lm}  \boldsymbol{\Psi}_j^{-1} \vect{D}_{\boldsymbol{\delta}(t)}^{\Htran} \tilde{\vect{x}}_{jk} + N(N\!-\!1) \notag \\
&\!\!\!\!\!\!\!\!\!\!\!\!\!\!\!\!\!\!\!\!\!\!\!\!\!\!\!\!\!\!\!\!\!\!\!\!\!\!\!\!\!\!\!\!\!\!\!\!\!\!\!\!
\times\begin{cases}
\lambda_{jjk}^2 \lambda_{jlm}^2 \tilde{\vect{x}}_{jk}^{\Htran}   \vect{D}_{\boldsymbol{\delta}(t)} \boldsymbol{\Psi}_j^{-1} \bar{\vect{X}}_{lm}  \boldsymbol{\Psi}_j^{-1} \vect{D}_{\boldsymbol{\delta}(t)}^{\Htran} \tilde{\vect{x}}_{jk}  & \!\! \textrm{if a CLO} \\
 \lambda_{jjk}^2 \lambda_{jlm}^2 |\tilde{\vect{x}}_{jk}^{\Htran}  \vect{D}_{\boldsymbol{\delta}(t)} \boldsymbol{\Psi}_j^{-1}  \vect{D}_{\boldsymbol{\delta}(t)}^{\Htran} \tilde{\vect{x}}_{lm} |^2 & \!\! \textrm{if SLOs}
\end{cases} \notag \\
\mathbb{E}\{ |\vect{v}_{jk}^{\Htran}(t) \boldsymbol{\upsilon}_j(t) |^2  \}&= \kappa^2 \mathbb{E}\{ \|\vect{v}_{jk}(t) \|^2\} \sum_{l=1}^{L} \sum_{m=1}^{K} p_{lm} \lambda_{jlm}  \notag \\
&\!\!\!\!\!\!\!\!\!\!\!\!\!\!\!\!\!\!\!\!\!\!\!\!\!\!\!\!\!\!\!\!\!\!\!\!\!\!\!\!\!\!\!\!\!\!\!\!\!\!\!\!\! +\!\! \kappa^2 \!\sum_{l=1}^{L} \sum_{m=1}^{K} p_{lm} N  \lambda_{jjk}^2 \lambda_{jlm}^2 \tilde{\vect{x}}_{jk}^{\Htran}   \vect{D}_{\boldsymbol{\delta}(t)} \boldsymbol{\Psi}_j^{-1} \vect{X}_{lm}  \boldsymbol{\Psi}_j^{-1} \vect{D}_{\boldsymbol{\delta}(t)}^{\Htran} \tilde{\vect{x}}_{jk}. \notag
\end{align}
\end{lemma}

The expectations for SLOs were previously derived in \cite[Theorem 2]{Bjornson2014c} and, interestingly, it is only the second order moments that are different with a CLO. Hence, the case with the smallest variance  $\sum_{l=1}^{L} \sum_{m=1}^{K} p_{lm}  \mathbb{E}\{ |\vect{v}_{jk}^{\Htran}(t) \vect{h}_{jlm}(t) |^2  \}$ of interference gives the largest rate for UE $k$ in cell $j$. This term depends mainly on the pilot sequences and phase drifts, as seen from the expressions in Lemma \ref{theorem:MRC-expectations}. The only difference is that $\bar{\vect{X}}_{\ell m}$ with a CLO is replaced by $\vect{D}_{\boldsymbol{\delta}(t)}  \tilde{\vect{x}}_{\ell m} \tilde{\vect{x}}_{\ell m}^{\Htran} \vect{D}_{\boldsymbol{\delta}(t)}^{\Htran} $ with SLOs. These terms are equal when there are no phase drifts (i.e., $\delta=0$), while the difference grows larger with $\delta$. In particular, the term $\bar{\vect{X}}_{\ell m}$ is unaffected by the time index $t$, while the corresponding term for SLOs decays as $e^{-\delta t}$ (from $\vect{D}_{\boldsymbol{\delta}(t)}$).
Hence, we expect SLOs to provide larger user rates than a CLO, because interference reduces faster with $t$ when the independent phase drifts mitigate pilot contamination.

By letting $N \rightarrow \infty$ in Lemma \ref{theorem:MRC-expectations}, one can obtain closed-form expressions for the asymptotic SINRs. It can be seen that the detrimental impact of hardware imperfections vanishes almost completely as $N$ grows large \cite[Corollary 1]{Bjornson2014c}. This result holds for any fixed values of the parameters $\delta$, $\kappa$, and $\xi$. It is also possible to increase these parameters with $N$. This gives a gradual degradation of the circuits' quality at the BS and the scaling should fulfill the following scaling law.

\begin{lemma}\label{lemma:scaling-law}
Suppose the hardware imperfection parameters are replaced as $\kappa^2 \mapsto \kappa_{0}^2 N^{\tau_1}$, $\xi \mapsto \xi_{0} N^{\tau_2}$, and $\delta \mapsto \delta_{0} (1+ \log_e(N^{\tau_3}) )$, for some scaling parameters $\tau_1,\tau_2,\tau_3 \geq 0$ and some initial values $\kappa_0,\xi_0,\delta_0 \geq 0$. If
\begin{equation} \label{eq:scaling-law}
\begin{cases} \max(\tau_1,\tau_2) \leq \frac{1}{2} \,\,\, \textrm{and} \,\,\, \tau_3=0 & \textrm{if a CLO} \\
\max(\tau_1,\tau_2) + \frac{\delta_{0} (t-B)}{2}\tau_3 \leq \frac{1}{2} & \textrm{if SLOs}.
\end{cases}
\end{equation}
then $\mathrm{SINR}_{jk}(t)$ with MRC has a non-zero limit as $N \! \rightarrow \! \infty$.
\end{lemma}
\begin{proof}
This follows along the lines of \cite[Corollary 3]{Bjornson2014c}.
\end{proof}

Lemma \ref{lemma:scaling-law} proves that the circuit design can be relaxed as $N$ increases. By accepting larger distortions we can achieve better energy efficiency in the circuits or lower hardware costs; see Section \ref{sec:circuit-examples}. The scaling law shows that the variances of the additive distortion noise and noise amplification can be increased simultaneously as $N$ to some exponent. The phase-drift variance can scale only for SLOs and only logarithmically with $N$, since it affects the signal itself. In this case, \eqref{eq:scaling-law} manifests a trade-off between increasing imperfections that cause additive and multiplicative distortions.

\section{Scaling Law Aware Circuit Design}
\label{sec:circuit-examples}

We now exemplify what the scaling law in Lemma
\ref{lemma:scaling-law} means for the hardware components at the BS, depicted in Fig.~\ref{figure_blockdiagram}.

\subsection{Analog-to-Digital Converter (ADC)}

The ADC quantizes the received signal to $b$ bit resolution. The quantization error can be included in the additive distortion noise $\boldsymbol{\upsilon}_{j}(t)$ and contributes to $\kappa^2$ with $2^{-2b}$
 \cite{Mezghani2010a}. The scaling law in Lemma \ref{lemma:scaling-law} allows us to increase the variance $\kappa^2$ as $N^{\tau_1}$ for $\tau_1 \leq \frac{1}{2}$. This corresponds to reducing the resolution of the ADC with $\frac{\tau_1}{2} \log_2(N)$ bits, which allows for substantial cost reduction. For example, we can reduce the ADC resolution by 2 bits if we deploy 256 antennas instead of one. For very large arrays, it is even sufficient to use 1-bit ADCs.

The power dissipation of an ADC, $P_{\mathrm{ADC}}$, is proportional to $2^{2b}$ \cite[Eq.~(8)]{Mezghani2010a} and can, thus, be decreased as $1/N^{\tau_1}$. If each antenna has a separate ADC, the total power $N P_{\mathrm{ADC}}$ still increases with $N$ but proportionally to $N^{1-\tau_1}$ for  $\tau_1 \leq \frac{1}{2}$, instead of $N$, due to the gradually lower ADC resolution.

\vspace{-1mm}

\subsection{Low Noise Amplifier (LNA)}

The LNA is an analog circuit that amplifies the received signal. It is shown in \cite{Song2008b} that the behavior of an LNA is characterized by the figure-of-merit (FoM) expression \vspace{-2mm}
\begin{equation} \label{eq:FOM-LNA}
\mathrm{FoM}_{\mathrm{LNA}} = \frac{G}{(\xi -1) P_{\mathrm{LNA}}}
\end{equation} \vskip-2mm
\noindent where $\xi$ is the noise amplification factor defined in Section \ref{sec:system-model}, $G$ is the amplifier gain, and $P_{\mathrm{LNA}}$ is the power consumption of the LNA. For optimized LNAs, $\mathrm{FoM}_{\mathrm{LNA}}$ is a constant determined by the circuit architecture \cite{Song2008b}; thus, $\mathrm{FoM}_{\mathrm{LNA}}$ basically scales with the hardware cost.  The scaling law in Lemma \ref{lemma:scaling-law} allows us to increase $\xi$ proportional to $N^{\tau_2}$ for $\tau_2 \leq \frac{1}{2}$. The noise figure, defined as $10 \log_{10} (\xi)$, can thus be increased by $\tau_2 10 \log_{10} (N)$ dB. For example, we can increase it by 10 dB if we deploy 100 antennas instead of one.

For a given architecture, the invariance of the $\mathrm{FoM}_{\mathrm{LNA}}$ in \eqref{eq:FOM-LNA} implies that we can decrease the power consumption (roughly) proportional to $1/N^{\tau_2}$. Hence, we can make the total power consumption of the $N$ LNAs, $N P_{\mathrm{LNA}}$, increase as $N^{1-\tau_2}$ instead of $N$ by increasing the noise amplification.

\vspace{-1mm}

\subsection{Local Oscillator (LO)}

Phase noise in the LOs is the main source of multiplicative phase drifts. If the LOs are free-running, the drifts are modeled by the Wiener process, defined in Section \ref{sec:system-model}, with variance
\begin{equation} \label{eq:oscillator-variance}
\delta = 4\pi^2 f_c^2 T_s \zeta
\end{equation}
where $f_c$ is the carrier frequency, $T_s$ is the symbol time, and $\zeta$ is a constant that characterizes the quality of the LO  \cite{Petrovic2007a}. Moreover, the power dissipation $P_{\mathrm{LO}}$ in an LO is directly coupled to $\zeta$, such that $P_{\mathrm{LO}} \zeta \approx \mathrm{FoM}_{\mathrm{LO}}$ where the FoM value $\mathrm{FoM}_{\mathrm{LO}}$ depends on the circuit architecture \cite{Petrovic2007a,Park2008a} and naturally on the hardware cost. For a given architecture, we can increase $\delta$ and, thereby, decrease the power $P_{\mathrm{LO}}$. The scaling law in Lemma \ref{lemma:scaling-law} allows us to increase $\delta$ proportionally to $(1+ \log_e(N^{\tau_3}) )$ when using SLOs. Hence, the power dissipation in the LOs can be reduced as $\frac{1}{1+\tau_3 \log_e(N)}$. This reduction is only logarithmic in $N$, which stands in contrast to the $1/\sqrt{N}$ scalings for ADCs and LNAs (with $\tau_1\!=\!\tau_2 \!=\! \frac{1}{2}$). Since linear increase is much faster than logarithmic decay, the total power $N P_{\mathrm{LO}}$ with SLOs increases almost linearly with $N$. Note that no scalings are allowed when having a CLO.

The LO variance formula in \eqref{eq:oscillator-variance} gives other possibilities than decreasing the circuit power. In particular, one can increase the carrier frequency $f_c$ with $N$ by exploiting the scaling law in Lemma \ref{lemma:scaling-law}. This is interesting because massive MIMO has been identified as a key enabler for operating in millimeter wave bands, in which phase noise is more severe since the variance in \eqref{eq:oscillator-variance} increases as $f_c^2$. Fortunately, massive MIMO has an inherent resilience towards phase noise.

\vspace{-3mm}

\section{Numerical Example}

\vspace{-1mm}

The analytic results are corroborated by simulating a scenario with 16 cells and wrap-around to avoid edge effects; see Fig.~\ref{figure_simulationscenario}. Each square cell is $250 \times  250$ meters and divided into 8 virtual sectors; each sector contains one uniformly distributed UE (with minimum distance $35$ meters). Each sector has an orthogonal pilot sequence from a DFT matrix \cite{Bjornson2014c}, but the same pilot is reused in the same sector of other cells.

The channel attenuations are $\lambda_{jlk} = 10^{s_{jlk}-1.53}/d_{jlk}^{3.76}$ where $d_{jlk}$ is the distance in meters between BS $j$ and UE $k$ in cell $l$ and $s_{jlk} \sim \mathcal{N}(0,0.25)$ is a realization of the shadow-fading. The transmit powers are $p_{jk} = -47$ dBm/Hz, the thermal noise power is $\sigma^2 = -174$ dBm/Hz, $B=8$ is the pilot sequence length, and the coherence block is $T=500$.

The achievable sum rate is shown in Fig.~\ref{figure_simulation1} for a system with either ideal hardware or hardware imperfections given by $b=8$ bit ADCs, 2 dB noise figure in the LNAs, and LOs with a phase noise variance of $1.6 \cdot 10^{-4}$. This corresponds to $\{ \kappa_0, \, \xi_0, \, \delta_0 \} \!=\! \{2^{-8}, \, 10^{0.2}, \, 1.6 \cdot 10^{-4}\}$ when we scale the hardware imperfections with $N$ as described in Lemma \ref{lemma:scaling-law}.

We see that the throughput is reduced by hardware imperfections. The loss is larger when the imperfections are increased with $N$, but the difference essentially vanishes as $N \rightarrow \infty$ if the scaling law for SLOs in Lemma \ref{lemma:scaling-law} is satisfied. The throughput loss is large when the scaling law is not followed.
We observe that SLOs provide higher throughput than a CLO. This is because parts of the interference average out.

\vspace{-3mm}

\section{Conclusion}

\vspace{-1mm}

Massive MIMO systems are prone to hardware imperfections in ADCs, LNAs, and LOs. We have shown that these systems have an inherent resilience to such imperfections. The distortions can be increased with $N$, which allows the circuit power of ADCs and LNAs to increase as $\sqrt{N}$ instead of $N$. The analysis shows that having a CLO is better in terms of energy efficiency and cost, while SLOs provide higher throughput.

\begin{figure}
\begin{center}
\includegraphics[width=\columnwidth]{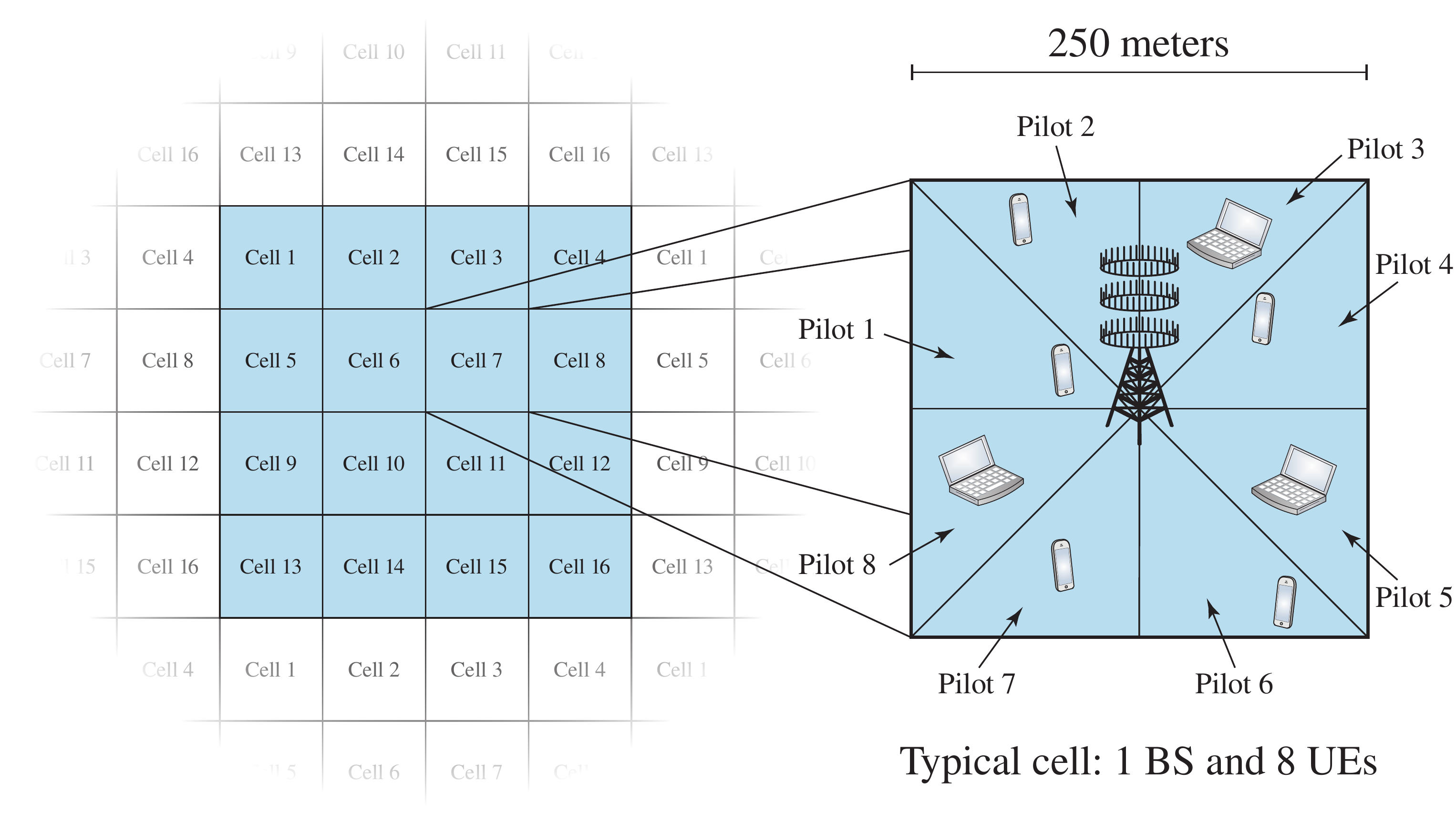}
\end{center}\vskip-6mm
\caption{The simulation scenario consists of 16 square cells with wrap-around to avoid edge effects.} \label{figure_simulationscenario}
\end{figure}

\begin{figure}[t!]
\begin{center}
\includegraphics[width=\columnwidth]{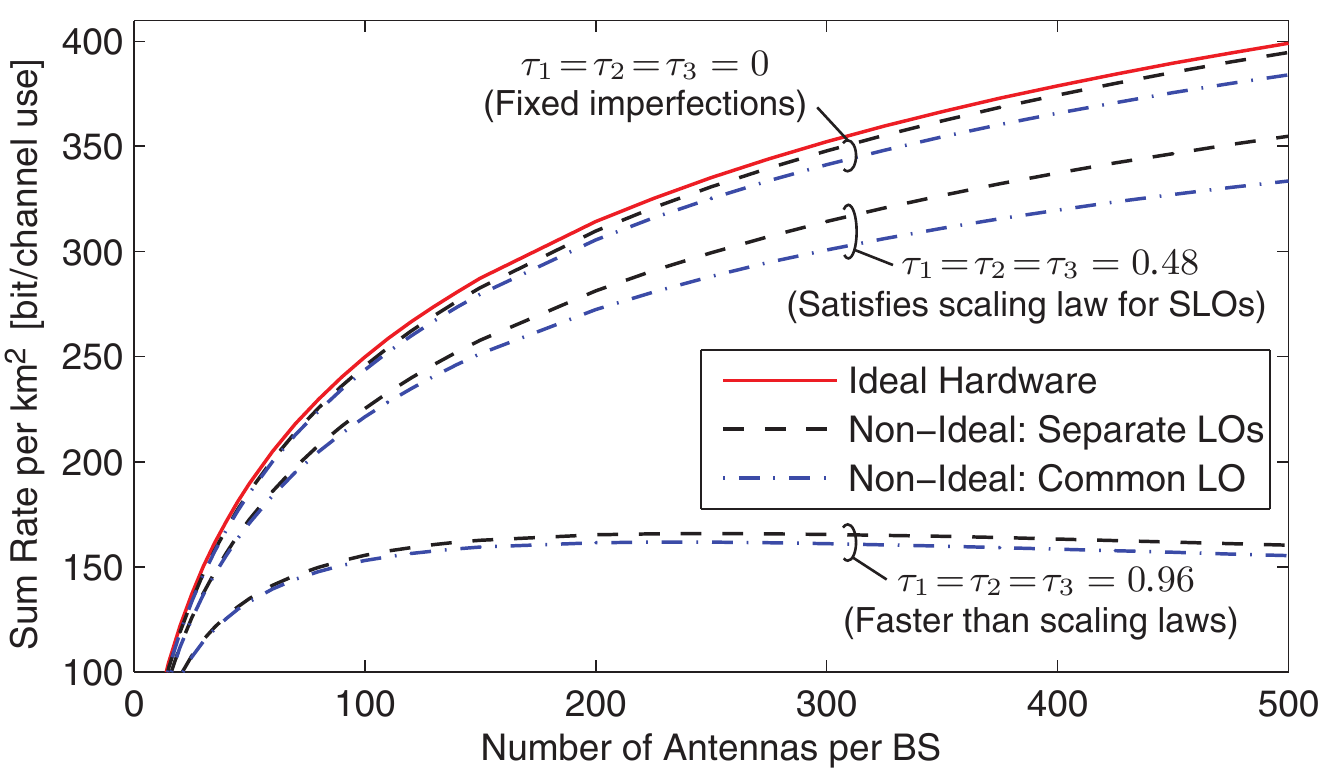}
\end{center}\vskip-7mm
\caption{Sum rate in the $1 \, \mathrm{km}^2$ area in Fig.~\ref{figure_simulationscenario} as a function of the numbers of antennas. We consider different hardware imperfections and scaling strategies.} \label{figure_simulation1} \vskip-1mm
\end{figure}

\vspace{-2mm}

\ninept

\bibliographystyle{IEEEbib}
\bibliography{IEEEabrv,refs}

\end{document}